\documentclass[runningheads,envcountsame]{llncs}
\usepackage{microtype}
\usepackage{mysty}

\title{Behavioral Specification Theories: \\ an Algebraic Taxonomy}

\author{Uli Fahrenberg\inst1 \and Axel Legay\inst2}

\institute{{\'E}cole polytechnique, Palaiseau, France \and
  Universit{\'e} Catholique de Louvain, Belgium}

\begin{document}

\maketitle

\begin{abstract}
  We develop a taxonomy of different behavioral specification theories
  and expose their algebraic properties.  We start by clarifying what
  precisely constitutes a behavioral specification theory and then
  introduce logical and structural operations and develop the
  resulting algebraic properties.  In order to motivate our
  developments, we give plenty of examples of behavioral specification
  theories with different operations.
\end{abstract}

\section{Introduction}

Behavioral specification theories are specification formalisms for
formal models which are enriched with logical and structural
operations.  This allows for incremental and compositional design and
verification and has shown itself to be a viable way to avoid the
habitual state-space explosion problems associated with the
verification of complex models.

Behavioral specification theories have seen significant attention in
recent years \cite{DBLP:journals/scp/AcetoFFIP13,
  DBLP:conf/fase/BauerDHLLNW12, DBLP:conf/atva/BenesCK11,
  DBLP:conf/avmfss/Larsen89, DBLP:conf/concur/Larsen90,
  DBLP:journals/iandc/BenesFKLT20, DBLP:conf/concur/BenesDFKL13,
  DBLP:conf/lics/LarsenX90, DBLP:journals/entcs/Raclet08,
  DBLP:journals/eatcs/AntonikHLNW08, DBLP:journals/tecs/BujtorSV17,
  DBLP:journals/tecs/BujtorV15, DBLP:conf/ictac/CaillaudR12}.
Generally speaking, they have the property that
the specification formalism is an extension of the modeling formalism,
so that specifications have an operational interpretation and models
are verified by comparing their operational behavior against the
specification's behavior.

Popular examples of behavioral specification theories are modal
transition systems~\cite{DBLP:journals/eatcs/AntonikHLNW08,
  DBLP:journals/tecs/BujtorV15, DBLP:conf/avmfss/Larsen89},
disjunctive modal transition
systems~\cite{DBLP:journals/tecs/BujtorSV17,
  DBLP:conf/concur/BenesDFKL13, DBLP:conf/ictac/FahrenbergLT14,
  DBLP:journals/iandc/BenesFKLT20, DBLP:conf/atva/BenesCK11,
  DBLP:journals/jlap/FahrenbergL20, DBLP:conf/sofsem/FahrenbergL17,
  DBLP:conf/lics/LarsenX90}, and acceptance
automata~\cite{DBLP:journals/entcs/Raclet08,
  DBLP:conf/ictac/CaillaudR12}.  Also relations to contracts and
interfaces have been exposed~\cite{DBLP:conf/fase/BauerDHLLNW12,
  DBLP:journals/fuin/RacletBBCLP11}, as have extensions for real-time,
probabilistic, and quantitative specifications and for models with
data~\cite{DBLP:journals/mscs/BauerJLLS12,
  DBLP:journals/soco/FahrenbergKLT18, conf/fit/FahrenbergL12,
  journals/lmcs/DelahayeFLL14, DBLP:conf/qest/DelahayeFLL13,
  DBLP:conf/mfcs/BauerFJLLT11, DBLP:journals/scp/BertrandLPR12,
  DBLP:journals/sttt/DavidLLNTW15, DBLP:journals/fmsd/BauerFJLLT13,
  DBLP:journals/acta/FahrenbergL14}.

Except for the work by Vogler~\etal
in~\cite{DBLP:journals/tecs/BujtorSV17, DBLP:journals/tecs/BujtorV15}
and our own~\cite{DBLP:journals/jlap/FahrenbergL20}, behavioral
specification theories have been developed only to characterize
bisimilarity (or variants like timed or probabilistic bisimilarity).
While bisimilarity is an important equivalence relation on models,
there are many others which also are of interest.  Examples include
nested and $k$-nested simulation~\cite{DBLP:journals/iandc/GrooteV92,
  DBLP:journals/iandc/AcetoFGI04}, ready or
$\tfrac23$-simulation~\cite{DBLP:conf/popl/LarsenS89}, trace
equivalence~\cite{DBLP:journals/cacm/Hoare78}, impossible
futures~\cite{DBLP:books/sp/Vogler92}, or the failure semantics
of~\cite{DBLP:journals/tecs/BujtorV15, DBLP:journals/tecs/BujtorSV17,
  DBLP:journals/jacm/BrookesHR84, DBLP:journals/acta/Vogler89,
  DBLP:conf/icalp/Pnueli85} and others.  We have addressed some of
these equivalences in~\cite{DBLP:journals/jlap/FahrenbergL20}.

In this survey we take a step back and develop a \emph{systemization}
or \emph{taxonomy} of different behavioral specification theories and
expose their \emph{algebraic} properties.  As an example, the most
basic ingredient of a behavioral specification theory is a preorder of
\emph{refinement} on specifications, turning the set of specifications
into a partial order up to~$\equiv$, the equivalence generated by
refinement.  Now if the refinement preorder admits least upper bounds,
then this binary operation is usually called \emph{conjunction}, and
the set of specifications becomes a meet-semilattice up to~$\equiv$.
Conjunction is a useful ingredient of any specification theory, but
some also admit \emph{disjunctions}, thus turning them into
\emph{distributive lattices} up to~$\equiv$.

We believe that a systemization as we set out for here is useful to
clarify which properties one needs or expects of behavioral
specification theories, and that it may help in developing new
behavioral specification theories, both for equivalence relations
different from bisimilarity and for more intricate models such as
real-time, probabilistic, or hybrid systems.

To develop our systemization, we first have to clarify what precisely
\emph{is} a behavioral specification theory.  Here we follow the
seminal work of Pnueli~\cite{DBLP:conf/icalp/Pnueli85}, Hennessy and
Milner~\cite{DBLP:journals/jacm/HennessyM85}, and
Larsen~\cite{DBLP:conf/concur/Larsen90} and argue that a behavioral
specification theory is built on an \emph{adequate} and
\emph{expressive} specification formalism equipped with a mapping from
models to their \emph{characteristic formulae}, which provides the
extension of the modeling formalism by the specification formalism.
This is the theme of Sections~2 and~3.

Section~4 then introduces behavioral specification theories, and
Section~5 makes precise what it means to have logical operations on
specifications.  Section~6 is concerned with \emph{structural}
operations on specifications: composition and quotient.  When present
in a specification theory, these can be used for compositional design
and verification.  Algebraically, a specification theory which has all
the logical and structural operations forms a \emph{residuated
  lattice} up to~$\equiv$, a well-understood algebraic
structure~\cite{JipsenT02} which also appears in linear
logic~\cite{DBLP:journals/tcs/Girard87} and other areas.

All throughout Sections~2 to~6, we give plenty of examples, taken from
our own work in~\cite{DBLP:journals/jlap/FahrenbergL20,
  DBLP:journals/iandc/BenesFKLT20,
  DBLP:journals/soco/FahrenbergKLT18}, of specification theories which
have the required operations.
In the final Section~7 we survey a few other behavioral specification
theories, for real-time and probabilistic models, in order to expose
their particular algebraic properties.  We make no claim to
completeness of this survey; indeed there are many other examples
which we do not treat here.  The paper finishes with a scheme which
sums up the relevant algebraic structures and an overview of the
properties of the different behavioral specification theories
encountered.

\section{Models and Specifications}

Let $\Spec$ be a set of \emph{specifications}, $\Proc$ a set of
\emph{models}, $\mathord{\models}\subseteq \Proc\times \Spec$ a
relation between specifications and models, and
$\mathord{\sim}\subseteq \Proc\times \Proc$ an equivalence relation on
$\Proc$.  The intuition is that $\Spec$ is to provide specifications
for the models in $\Proc$ through the relation $\models$, but \emph{up
  to $\sim$}, so that two models which are equivalent cannot be
distinguished by their specifications.  We will make this precise
below.

We will generally use $\mcal S$ for specifications and $\mcal M$ for
models.  For $\mcal S\in \Spec$, let
$\Mod{ \mcal S}=\{ \mcal M\in \Proc\mid \mcal M\models \mcal S\}$
denote its \emph{set of implementations}.  For $\mcal M\in \Proc$, let
$\Th{ \mcal M}=\{ \mcal S\in \Spec\mid \mcal M\models \mcal S\}$
denote its \emph{set of theories}.  We record the following trivial
fact:

\begin{samepage}
\begin{lemma}
  For any $\mcal S\in \Spec$ and $\mcal M\in \Proc$, the following are
  equivalent:
  \begin{enumerate}[(1)]
  \item $\mcal M\models \mcal S$;
  \item $\mcal M\in \Mod{ \mcal S}$;
  \item $\mcal S\in \Th{ \mcal M}$.
  \end{enumerate}
\end{lemma}
\end{samepage}

\begin{blexample}
  A common type of models is given by \emph{labeled transition
    systems} ($\LTS$).  These are structures $\mcal M=( S, s^0, T)$
  consisting of a finite set of states $S$, an initial state
  $s^0\in S$, and transitions $T\subseteq S\times \Sigma\times S$
  labeled with symbols from a fixed finite set $\Sigma$.

  $\LTS$ are often considered modulo \emph{bisimilarity}: A
  \emph{bisimulation} between two $\LTS$ $\mcal M_1=( S_1, s^0_1, T_1)$
  and $\mcal M_2=( S_2, s^0_2, T_2)$ is a relation
  $R\subseteq S_1\times S_2$ such that $( s^0_1, s^0_2)\in R$ and for
  any $( s_1, s_2)\in R$,
  \begin{enumerate}[(1)]
  \item for all $( s_1, a, s_1')\in T_1$ there exists $( s_2, a,
    s_2')\in T_2$ such that $( s_1', s_2')\in R$;
  \item for all $( s_2, a, s_2')\in T_2$ there exists
    $( s_1, a, s_1')\in T_1$ such that $( s_1', s_2')\in R$;
  \end{enumerate}
  and then $\mcal M_1$ and $\mcal M_2$ are said to be \emph{bisimilar}
  if there exists a bisimulation between them.

  A common specification formalism for $\LTS$ is \emph{Hennessy-Milner
    logic}~\cite{DBLP:journals/jacm/HennessyM85}.  It consists of
  formulae generated by the abstract syntax
  \begin{equation*}
    \HML\ni \phi, \psi\Coloneqq \ltrue\mid \lfalse\mid \phi\land
    \psi\mid \phi\lor \psi\mid \langle a\rangle \phi\mid[ a] \phi \quad(
    a\in \Sigma)\,,
  \end{equation*}
  with semantics defined by $\Mod \ltrue= \LTS$,
  $\Mod \lfalse= \emptyset$,
  $\Mod{ \phi\land \psi}= \Mod \phi\cap \Mod \psi$,
  $\Mod{ \phi\lor \psi}= \Mod \phi\cup \Mod \psi$, and
  \begin{align*}
    &\Mod{ \langle a\rangle \phi}=\{( S, s^0, T)\in \LTS\mid \exists(
    s^0, a, s)\in T:( S, s, T)\in \Mod \phi\}\,; \\
    &\Mod{ [ a] \phi}=\{( S, s^0, T)\in \LTS\mid \forall( s^0, a,
    s)\in T:( S, s, T)\in \Mod \phi\}\,.
  \end{align*}
  The Hennessy-Milner theorem~\cite{DBLP:journals/jacm/HennessyM85}
  then states that $\HML$ specifies $\LTS$ up to bisimilarity, that
  is, $\mcal M_1\sim \mcal M_2$ precisely when
  $\Th{ \mcal M_1}= \Th{ \mcal M_2}$. \qed
\end{blexample}

\begin{definition}[\cite{DBLP:journals/jacm/HennessyM85}]
  $( \Spec, \mathord{\models})$ is \emph{adequate} for
  $( \Proc, \mathord{\sim})$ if it holds for any
  $\mcal M_1, \mcal M_2\in \Proc$ that $\mcal M_1\sim \mcal M_2$ iff
  $\Th{ \mcal M_1}= \Th{ \mcal M_2}$.
\end{definition}

\section{Characteristic Formulae}

Let $\mcal M\in \Proc$.  A specification $\mcal S\in \Spec$ is a
\emph{characteristic formula for
  $\mcal M$}~\cite{DBLP:conf/icalp/Pnueli85} if it holds for any
$\mcal M'\in \Proc$ that $\mcal M'\models \mcal S$ iff
$\Th{ \mcal M'}= \Th{ \mcal M}$.

\begin{lemma}
  If $\mcal S_1, \mcal S_2\in \Spec$ are characteristic formulae for
  $\mcal M\in \Proc$, then $\Mod{ \mcal S_1}= \Mod{ \mcal S_2}$.
\end{lemma}

\begin{proof}
  For any $\mcal M'\in \Proc$, $\mcal M'\in \Mod{ \mcal S_1}$ iff
  $\Th{ \mcal M'}= \Th{ \mcal M}$, iff $\mcal M'\in \Mod{ \mcal
    S_2}$. \qed
\end{proof}

\begin{definition}[\cite{DBLP:conf/icalp/Pnueli85}]
  $( \Spec, \mathord{\models})$ is \emph{expressive} for
  $( \Proc, \mathord{\sim})$ if every $\mcal M\in \Proc$ admits a
  characteristic formula.
\end{definition}

\begin{blexample}
  It is known~\cite{books/AcetoILS07} that $\HML$ is \emph{not}
  expressive for $\LTS$ with bisimilarity.  Indeed, the simple
  transition system $(\{ s^0\}, s^0,( s^0, a, s^0)\}$ consisting only
  of a loop at the initial state does not admit a characteristic
  formula in $\HML$.

  The standard remedy~\cite{DBLP:journals/tcs/Larsen90} for this
  expressivity failure is to add recursion and maximal fixed points to
  the logic.  For a finite set $X$ of variables, let $\HML( X)$ be the
  set of formulae generated as follows:
  \begin{equation*}
    \HML( X)\ni \phi, \psi\Coloneqq \ltrue\mid \lfalse\mid \phi\land
    \psi\mid \phi\lor \psi\mid \langle a\rangle \phi\mid[ a] \phi \mid x
    \quad( a\in \Sigma, x\in X)
  \end{equation*}
  That is, $\HML( X)$ formulae are $\HML$ formulae which additionally
  may contain variables from $X$.

  A \emph{recursive Hennessy-Milner
    formula}~\cite{DBLP:journals/iandc/BenesFKLT20,
    DBLP:journals/tcs/Larsen90} is a tuple $\mcal H=( X, X^0, \Delta)$
  consisting of finite sets $X\supseteq X^0$ of \emph{variables} and
  \emph{initial} variables, respectively, and a \emph{declaration}
  $\Delta: X\to \HML( X)$.  The set of such formulae is denoted
  $\HMLR$.  The semantics of a formula $\mcal H\in \HMLR$ is a set
  $\sem{ \mcal H}\in \LTS$ which is defined as a maximal fixed point,
  see~\cite{DBLP:journals/tcs/Larsen90, books/AcetoILS07}; we do not
  go into these details here because we will give another, equivalent,
  semantics below.

  The \emph{characteristic formula}~\cite{DBLP:journals/tcs/Larsen90}
  of $( S, s^0, T)\in \LTS$ is now the $\HMLR$ formula
  $( S,\{ s^0\}, \Delta)$ given by
  \begin{equation*}
    \Delta( s)= \bigland_{( s, a, t)\in T} \langle a\rangle t\land
    \bigland_{ a\in \Sigma}[ a]\Big( \biglor_{( s, a, t)\in T)}
    t\Big)\,.
  \end{equation*}
  Note how $\Delta( s)$ precisely specifies all labels which must be
  available from $s$ (the first part of the conjunction) and, for each
  label, which properties must be satisfied after its occurrence (the
  second part of the conjunction).  \qed
\end{blexample}

\section{Specification Theories}

\begin{definition}[\cite{DBLP:journals/jlap/FahrenbergL20}]
  A \emph{behavioral specification theory} for
  $( \Proc, \mathord{\sim})$ consists of a~set\/ $\Spec$ of
  specifications, a relation
  $\mathord{\models}\subseteq \Proc\times \Spec$, a~mapping
  $\chi: \Proc\to \Spec$, and a~preorder $\le$ on $\Spec$, called
  \emph{refinement}, subject to the following conditions:
  \begin{enumerate}[(1)]
  \item \label{en:spt.ad} $( \Spec, \mathord{\models})$ is adequate
    for $( \Proc, \mathord{\sim})$;
  \item \label{en:spt.ch} for every $\mcal M\in \Proc$,
    $\chi( \mcal M)$ is a characteristic formula for $\mcal M$;
  \item \label{en:spt.mr} for all $\mcal M\in \Proc$ and all
    $\mcal S\in \Spec$, $\mcal M\models \mcal S$ iff
    $\chi( \mcal M)\le \mcal S$.
  \end{enumerate}
\end{definition}

We will generally omit ``behavioral'' from now and only speak about
\emph{specification theories}.

\pagebreak

The equivalence relation $\mathord{ \equiv}$ on $\Spec$ defined as
$\mathord{ \le}\cap \mathord{ \ge}$ is called \emph{modal
  equivalence}.  Some comments on the different ingredients above are
in order.
\begin{enumerate}
\item By~\eqref{en:spt.ch}, $( \Spec, \mathord{\models})$ is also
  \emph{expressive} for $( \Proc, \mathord{\sim})$.
\item $\chi$ is a \emph{section} of $\models$: for all
  $\mcal M\in \Proc$, $\mcal M\models \chi( \mcal M)$.
\item \eqref{en:spt.mr} can be seen as \emph{defining} $\models$, so
  we may omit $\models$ from the signature of specification theories.
\item For any $\mcal M\in \Proc$,
  $\Th{ \mcal M}=\{ \mcal S\in \Spec\mid \chi( \mcal M)\le \mcal S\}=
  \chi( \mcal M)\mathord\uparrow$ is the \emph{upward closure} of
  $\chi( \mcal M)$ with respect to $\le$.
\end{enumerate}

\begin{lemma}[\cite{DBLP:journals/jlap/FahrenbergL20}]
  Let $( \Spec, \chi, \mathord{\le})$ be a specification theory for
  $( \Proc, \mathord{\sim})$.
  \begin{enumerate}[(1)]
  \item For all $\mcal S_1, \mcal S_2\in \Spec$,
    $\mcal S_1\le \mcal S_2$ implies $\Mod{ \mcal S_1}\subseteq \Mod{
      \mcal S_2}$.
  \item For all $\mcal M_1, \mcal M_2\in \Proc$,
    $\mcal M_1\sim \mcal M_2$ iff
    $\chi( \mcal M_1)\le \chi( \mcal M_2)$.
  \end{enumerate}
\end{lemma}

\begin{proof}
  For the first claim, $\mcal M\in \Mod{ \mcal S_1}$ implies
  $\chi( \mcal M)\le \mcal S_1\le \mcal S_2$, hence $\mcal M\in \Mod{
    \mcal S_2}$.

  For the second claim, we have $\mcal M_1\sim \mcal M_2$ iff
  $\mcal M_1\models \chi( \mcal M_2)$ (as $\chi( \mcal M_2)$ is
  characteristic for $\mcal M_2$), iff
  $\chi( \mcal M_1)\le \chi( \mcal M_2)$ by~\eqref{en:spt.mr}. \qed
\end{proof}

\begin{blexample}
  \cite{DBLP:journals/iandc/BenesFKLT20}~introduces a \emph{normal
    form} for $\HMLR$ formulae, showing that for any $\HMLR$ formula
  $\mcal H_1=( X_1, X^0_1, \Delta_1)$, there exists another formula
  $\mcal H_2=( X_2, X^0_2, \Delta_2)$ with
  $\sem{ \mcal H_1}= \sem{ \mcal H_2}$ and such that for any $x\in X_2$,
  $\Delta_2( x)$ is of the form
  \begin{equation*}
    \Delta_2( x)= \bigland_{ N\in
      \Diamond(x)} \Big( \biglor_{( a, y)\in N} \langle a\rangle y\Big)
    \land \bigland_{ a\in \Sigma}[ a]\Big( \biglor_{ y\in \Box^a( x)}
    y\Big)\,,
  \end{equation*}
  for finite sets $\Diamond( x)\subseteq 2^{ \Sigma\times X_2}$ and,
  for each $a\in \Sigma$, $\Box^a( x)\subseteq X_2$.  This may be seen
  as generalizing the characteristic formulae of $\HMLR$: the first
  part of the conjunction in $\Delta_2( x)$ specifies all labels which
  must be available, and the second part, which properties must be
  satisfied after each label's occurrence.

  A \emph{refinement}~\cite{DBLP:journals/iandc/BenesFKLT20} of two
  $\HMLR$ formulae $\mcal H_1=( X_1, X^0_1, \Delta_1)$ and
  $\mcal H_2=( X_2, X^0_2, \Delta_2)$ in normal form is a relation
  $R\subseteq X_1\times X_2$ such that for every $x^0_1\in X^0_1$
  there exists $x^0_2\in X^0_2$ for which $( x^0_1, x^0_2)\in R$, and
  for any $( x_1, x_2)\in R$,
  \begin{enumerate}[(1)]
  \item for all $N_2\in \Diamond_2( x_2)$ there is $N_1\in \Diamond_1(
    x_1)$ such that for each $( a, y_1)\in N_1$, there exists $( a,
    y_2)\in N_2$ with $( y_1, y_2)\in R$;
  \item for all $a\in \Sigma$ and every $y_1\in
    \Box^a_1( x_1)$, there is $y_2\in \Box^a_2( x_2)$ for which $( y_1,
    y_2)\in R$.
  \end{enumerate}
  Note how this corresponds to the intuition for the normal form above.

  Writing $\mcal H_1\le \mcal H_2$ whenever there exists a refinement
  as above, and denoting by $\chi( \mcal M)$ the characteristic
  formula of $\mcal M\in \LTS$ introduced in the previous example, it
  can be shown~\cite{DBLP:journals/iandc/BenesFKLT20} that
  $( \HMLR, \chi, \mathord{\le})$ is a specification theory for $\LTS$
  under bisimulation.  This also implies that the refinement semantics
  of $\HMLR$ agrees with the standard fixed-point
  semantics~\cite{DBLP:journals/tcs/Larsen90, books/AcetoILS07}.  \qed
\end{blexample}

\begin{blexample}
  \cite{DBLP:journals/iandc/BenesFKLT20}~exposes structural
  translations between $\HMLR$ and two other specification formalism:
  a generalization of the \emph{disjunctive modal transition systems}
  ($\DMTS$) introduced in~\cite{DBLP:conf/lics/LarsenX90} to multiple
  initial states, and a non-deterministic version of the
  \emph{acceptance automata} ($\NAA$)
  of~\cite{DBLP:journals/entcs/Raclet08, DBLP:conf/ictac/CaillaudR12}.
  This yields two other specification theories for $\LTS$ under
  bisimulation, one $\DMTS$-based and one based on (non-deterministic)
  acceptance automata.  \qed
\end{blexample}

\begin{blexample}
  \cite{DBLP:journals/jlap/FahrenbergL20}~introduces $\DMTS$-based
  specification theories for $( \LTS, \mathord{\cong})$, where $\cong$
  is any equivalence in van~Glabbeek's linear-time--branching-time
  spectrum~\cite{inbook/hpa/Glabbeek01}.  Using the translations
  mentioned in the previous example, these also give rise to
  $\HMLR$-based specification theories, and to specification theories
  based on acceptance automata, for all those equivalences. \qed
\end{blexample}

\section{Logical Operations on Specifications}

Behavioral specifications typically come equipped with logical
operations of conjunction and disjunction.  Recall that
$\mathord{ \equiv}$ is defined as $\mathord{ \le}\cap \mathord{ \ge}$.

\begin{definition}
  A specification theory $( \Spec, \chi, \mathord{\le})$ for
  $( \Proc, \mathord{\sim})$ is \emph{logical} if $( \Spec,
  \mathord{\le})$ forms a bounded distributive lattice up to
  $\equiv$.
\end{definition}

The above implies that $\Spec$ admits commutative and associative
binary operations $\lor$ of least upper bound and $\land$ of greatest
lower bound: disjunction and conjunction.  It also entails that there
is a bottom specification $\lfalse\in \Spec$, satisfying
$\Mod{ \lfalse}= \emptyset$, and a top specification
$\ltrue\in \Spec$, satisfying $\Mod{ \ltrue}= \Proc$.  We sum up the
properties of these operations:
\begin{gather}
  \label{eq:lub} \mcal S_1\lor \mcal S_2\le \mcal S_3
  \quad\text{iff}\quad \mcal
  S_1\le \mcal S_3 \text{ and } \mcal S_2\le \mcal S_3 \\
  \label{eq:glb} \mcal S_1\le \mcal S_2\land \mcal S_3
  \quad\text{iff}\quad \mcal
  S_1\le \mcal S_2 \text{ and } \mcal S_1\le \mcal S_3 \\
  \notag \mcal S_1\land( \mcal S_2\lor \mcal S_3)\equiv( \mcal
  S_1\land \mcal
  S_2)\lor( \mcal S_1\land \mcal S_3) \\
  \notag \mcal S_1\lor( \mcal S_2\land \mcal S_3)\equiv( \mcal S_1\lor
  \mcal
  S_2)\land( \mcal S_1\lor \mcal S_3) \\
  \notag \lfalse \land \mcal S\equiv \lfalse \qquad \ltrue \land \mcal
  S\equiv \mcal S \\
  \notag \lfalse \lor \mcal S\equiv \mcal S \qquad \ltrue \lor \mcal
  S\equiv \ltrue
\end{gather}
Note that the properties of least upper bound and greatest lower bound
in~\eqref{eq:lub} and~\eqref{eq:glb} above \emph{define} $\lor$ and
$\land$ uniquely: they are universal properties.

\begin{blexample}
  Hennessy-Milner logic has disjunction and conjunction as part of the
  syntax, and~\cite{DBLP:journals/iandc/BenesFKLT20} shows that on the
  specification theory $( \HMLR, \chi, \mathord{\le})$ from previous
  examples these are operations as above.  The disjunction of two
  $\HMLR$ formulae in normal form is again in normal form; for
  conjunction it may be defined directly on normal forms as follows:

  Let $\mcal H_1=( X_1, X^0_1, \Delta_1)$ and
  $\mcal H_2=( X_2, X^0_2, \Delta_2)$ be $\HMLR$ formulae in normal
  form and define $\mcal H=( X_1\times X_2, X^0_1\times X^1_1,
  \Delta)$ by $\Box^a(( x_1, x_2))= \Box^a_1( x_1)\land \Box^a_2( x_2)$
  for every $a\in \Sigma$ and $( x_1, x_2)\in X$ and
  \begin{multline*}
    \Diamond(( x_1, x_2))= \big\{ \{( a,( y_1, y_2))\mid( a, y_1)\in
    N_1, y_2\in \Box^a_2( x_2)\}\bigmid N_1\in \Diamond_2( x_1)\big\}
    \\
    \cup \big\{ \{( a,( y_1, y_2))\mid( a, y_2)\in N_2, y_1\in
    \Box^a_1( x_1)\}\bigmid N_2\in \Diamond_2( x_2)\big\}\,,
  \end{multline*}
  then $\mcal H\equiv \mcal H_1\land \mcal
  H_2$~\cite{DBLP:journals/iandc/BenesFKLT20}.

  Hence the three specification theories for
  $( \Proc, \mathord{\sim})$
  of~\cite{DBLP:journals/iandc/BenesFKLT20}: $\HMLR$, $\DMTS$, and
  $\NAA$, are all logical.  \qed
\end{blexample}

As a variation, some specification theories only admit conjunction and
no disjunction, thus forming a \emph{bounded meet-semilattice}.  We
call such specification theories \emph{semi-logical}.

\section{Structural Operations on Specifications}

Many behavioral specifications also admit structural operations of
\emph{composition}, denoted~$\|$, and \emph{quotient}, denoted~$\by$,
in order to enable compositional design and verification.

\begin{definition}
  A \emph{compositional specification theory} is a specification
  theory $( \Spec, \chi, \mathord{\le})$ for
  $( \Proc, \mathord{\sim})$ together with an operation $\|$ on
  $\Spec$ such that $( \Spec, \mathord{\|}, \mathord{\le})$ forms a
  commutative partially ordered semigroup up to $\equiv$.
\end{definition}

That is to say that the operation $\|$ is commutative and associative
and additionally satisfies the following monotonicity law:
\begin{equation*}
  \mcal S_1\le \mcal S_2\limpl \mcal S_1\| \mcal S_3\le \mcal S_2\|
  \mcal S_3
\end{equation*}
Contrary to the logical operations $\lor$ and $\land$, $\|$ is
\emph{not} defined uniquely; indeed a specification theory may admit
many different composition operations.

\begin{corollary}[Independent implementability]
  If $( \Spec, \mathord{\|}, \chi, \mathord{\le})$ is compositional,
  then $\mcal S_1\le \mcal S_3$ and $\mcal S_2\le \mcal S_4$ imply
  $\mcal S_1\| \mcal S_2\le \mcal S_3\| \mcal S_4$.
\end{corollary}

\begin{proof}
  By monotonicity, $\mcal S_1\| \mcal S_2\le \mcal S_3\| \mcal S_2\le
  \mcal S_3\| \mcal S_4$. \qed
\end{proof}

Note that independent implementability also \emph{implies} the
monotonicity law above.

If $( \Spec, \mathord{\|}, \chi, \mathord{\le})$ is
\emph{compositional and logical}, then it is called a
\emph{lattice-ordered semigroup} (up to $\equiv$) as an algebraic
structure; more precisely a bounded distributive lattice-ordered
commutative semigroup.  This entails that composition distributes over
disjunction:
\begin{equation*}
  \mcal S_1\|( \mcal S_2\lor \mcal S_2)\equiv \mcal S_1\| \mcal
  S_2\lor \mcal S_1\| \mcal S_3
\end{equation*}
Note that composition does \emph{not} necessarily distributed over
\emph{conjunction}.

If composition $\|$ also admits a unit $\one\in \Spec$ (up to
$\equiv$), \ie~such that $\mcal S\| \one\equiv \mcal S$ for all
$\mcal S\in \Spec$, then $( \Spec, \mathord{\|}, \chi, \mathord{\le})$
is said to be \emph{unital}, and ``semigroup'' is replaced by
``monoid'' above.

\begin{definition}
  A compositional specification theory
  $( \Spec, \mathord{\|}, \chi, \mathord{\le})$ for
  $( \Proc, \mathord{\sim})$ is \emph{complete} if\/
  $( \Spec, \mathord{\|}, \mathord{\le})$ forms a \emph{residuated}
  partially ordered commutative semigroup up to~$\equiv$.
\end{definition}

That is, the operation $\|$ admits a \emph{residual} $\by$, in our
context called \emph{quotient}, satisfying the following property:
\begin{equation}
  \label{eq:quotient}
  \mcal S_1\| \mcal S_2\le \mcal S_3 \liff \mcal S_2\le \mcal S_3\by
  \mcal S_1
\end{equation}
This property is again universal, so that $\by$ is uniquely defined by
$\|$.

If $( \Spec, \mathord{\|}, \chi, \mathord{\le})$ is also unital, then
it forms a \emph{residuated poset} up to~$\equiv$.  In that case, the
following holds for all $\mcal S_1, \mcal S_2\in \Spec$:
\begin{equation*}
  \mcal S_1\|( \one\by \mcal S_2)\le \mcal S_1\by \mcal S_2  
\end{equation*}
We refer to~\cite{JipsenT02} for a survey on residuated posets and the
residuated lattices we will encounter in a moment; we only highlight a
few properties here.

\begin{lemma}[\cite{JipsenT02}]
  The following hold in any complete compositional specification
  theory:
  \begin{align*}
    \mcal S_1\|( \mcal S_2\by \mcal S_3) &\le ( \mcal S_1\| \mcal
    S_2)\by \mcal S_3 &
    \mcal S_1\by \mcal S_2 &\le ( \mcal S_1\| \mcal S_3)\by( \mcal
    S_2\| \mcal S_3) \\
    ( \mcal S_1\by \mcal S_2)\|( \mcal S_2\by \mcal S_3) &\le \mcal
    S_1\by \mcal S_3 &
    ( \mcal S_1\by \mcal S_2)\by \mcal S_3 &\equiv ( \mcal S_1\by
    \mcal S_3)\by \mcal S_2 \\
    \mcal S_1\by( \mcal S_2\| \mcal S_3) &\equiv ( \mcal S_1\by \mcal
    S_2)\by \mcal S_3 &
    \mcal S\|( \mcal S\by \mcal S) &\equiv \mcal S \\
    ( \mcal S\by \mcal S)\|( \mcal S\by \mcal S) &\equiv \mcal S\by \mcal S
  \end{align*}
\end{lemma}

If $( \Spec, \mathord{\|}, \one, \chi, \mathord{\le})$ is
\emph{complete compositional and logical}, then it is called a
\emph{residuated lattice-ordered semigroup} (up to $\equiv$); more
precisely a bounded distributive residuated
lattice-ordered commutative semigroup.  Distributivity of composition over
disjunction now follows from residuation, and also the quotient is
well-behaved with respect to the logical operations:
\begin{gather*}
  ( \mcal S_1\land \mcal S_2)\by \mcal S_3\equiv \mcal S_1\by \mcal
  S_3\land \mcal S_2\by \mcal S_3 \qquad \mcal S_1\by( \mcal S_2\lor
  \mcal S_3)\equiv \mcal S_1\by \mcal S_2\land \mcal S_1\by \mcal S_3
\end{gather*}
Additionally, composition and quotient interact with the bottom and
top elements as follows:
\begin{equation*}
  \mcal S\| \lfalse\equiv \lfalse \qquad \mcal S\by \lfalse\equiv
  \ltrue \qquad \ltrue\by \mcal S\equiv \ltrue
\end{equation*}

Finally, if $( \Spec, \mathord{\|}, \one, \chi, \mathord{\le})$ is
complete compositional, unital, and logical, then it is called a
\emph{residuated lattice}.  We sum up the different algebraic
structures we have encountered in Fig.~\ref{fi:alg}.

\begin{figure}[tb]
  \centering
  \hspace*{-2ex}
  \begin{tikzpicture}[x=1.2cm]
    \tikzstyle{every node}=[font=\small,text badly centered]
    \begin{scope}
      \node (log) at (0,0) {logical};
      \node (com) at (-2,-1) {compositional\vphantom{g}};
      \node (cal) at (0,-2) {comp.\ \& log.};
      \node (uco) at (-3,-3) {\;\;unital compositional\vphantom{g}};
      \node (ucl) at (-1,-4) {\;\;\;uni.\ comp.\ \& log.};
      \node (cco) at (-2,-5) {complete comp.};
      \node (ucc) at (-3,-7) {uni.\ complete comp.};
      \node (ccl) at (0,-6) {complete comp.\ \& log.};
      \node (uccl) at (-1,-8) {uni.\ complete comp.\ \& log.};
    \end{scope}
    \begin{scope}[xshift=5cm]
      \node (log1) at (0,0) {b.d.\ lattice\vphantom{g}};
      \node (com1) at (-2,-1) {po.c.\ semigroup};
      \node (cal1) at (0,-2) {b.d.lo.c.\ semigroup};
      \node (uco1) at (-3,-3) {po.c.\ monoid};
      \node (ucl1) at (-1,-4) {b.d.lo.c.\ monoid\vphantom{g}};
      \node (cco1) at (-2,-5) {\;residuated po.c.\ semigroup};
      \node (ucc1) at (-3,-7) {\,c.\ residuated poset};
      \node (ccl1) at (0,-6) {b.d.\ residuated lo.c.\ semigroup};
      \node (uccl1) at (-1,-8) {b.d.c.\ residuated lattice\vphantom{g}};
    \end{scope}
    \foreach \n in {log, com, cal, uco, ucl, cco, ucc, ccl, uccl} {%
      \path (\n) edge[-, blue] (\n1);
    }
    \foreach \n/\m in {log/cal, com/cal, com/uco, uco/ucl, cal/ucl,
      com/cco, cco/ucc, uco/ucc, cal/ccl, cco/ccl, ucl/uccl, ccl/uccl,
      ucc/uccl} {%
      \path (\n) edge[<-, black!50] (\m);
      \path (\n1) edge[<-, black!50] (\m1);
    }
  \end{tikzpicture}
  \caption{Spectrum of specification theories and the corresponding
    algebraic structures.  Abbreviations: b.---bounded;
    d.---distributive; c.---commutative; po.---partially ordered;
    lo.---lattice-ordered 
    }
  \label{fi:alg}
\end{figure}
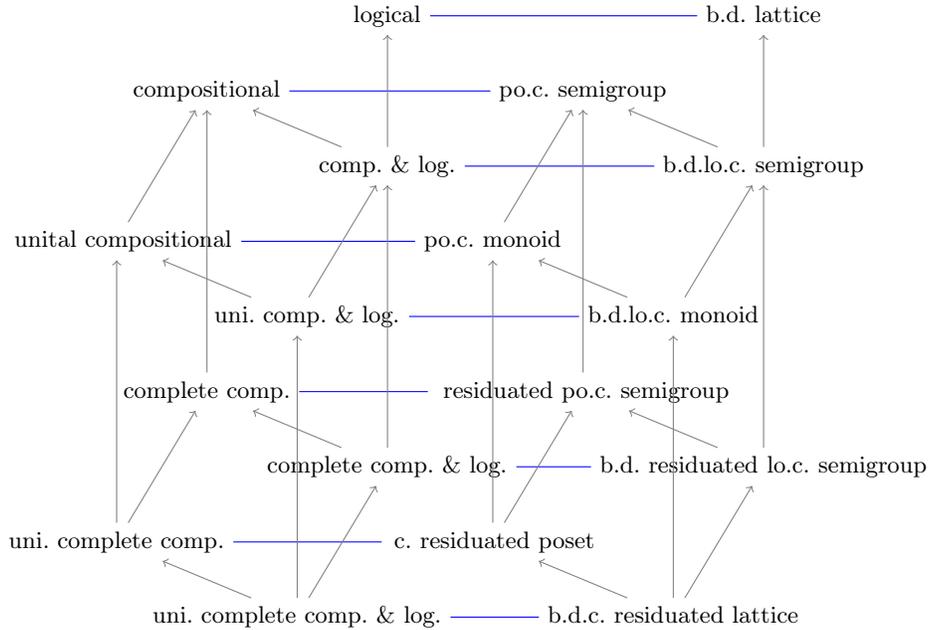

\begin{blexample}
  In~\cite{DBLP:journals/iandc/BenesFKLT20} it is shown that the
  specification theory $( \HMLR, \chi, \mathord{\le})$, and thus also
  the specification theories based on $\DMTS$ and $\NAA$, are unital
  complete compositional when enriched with CSP-style composition
  $\|$.  (In~\cite{DBLP:journals/soco/FahrenbergKLT18} this is
  generalized to other types of composition.)

  \begin{samepage}
    The composition $\mcal H_1\| \mcal H_2$ is defined by translation
    between $\HMLR$ and $\NAA$.  Also quotient is defined through
    $\NAA$, and it is shown in~\cite{DBLP:journals/iandc/BenesFKLT20}
    that due to these translations, composition may incur an
    exponential blow-up and quotient a double-exponential blow-up.
    \qed
  \end{samepage}
\end{blexample}

\section{Specification Theories for Real-Time and Probabilistic Systems}

We quickly survey a few different specification theories for real-time
and probabilistic systems.

\subsection{Modal event-clock specifications}

Modal event-clock specifications ($\MECS$) were introduced
in~\cite{DBLP:journals/scp/BertrandLPR12}.  They form a specification
theory for event-clock automata
($\ECA$)~\cite{DBLP:journals/tcs/AlurFH99}, a determinizable subclass
of timed automata~\cite{DBLP:journals/tcs/AlurD94}, under timed
bisimilarity.  Models and specifications are assume to be
deterministic, thus $\mcal S_1\le \mcal S_2$ iff
$\Mod{ \mcal S_1}\subseteq \Mod{ \mcal S_2}$ in this case.

In~\cite{DBLP:journals/scp/BertrandLPR12} it is shown that $\MECS$
admit a conjunction, thus forming a meet-semilattice up to~$\equiv$.
The authors also introduce composition and quotient; but computation
of quotient incurs an exponential blow-up.  Altogether, $\MECS$ form a
complete compositional semi-logical specification theory: a bounded
residuated semilattice-ordered commutative semigroup.

\subsection{Timed input/output automata}

\cite{DBLP:journals/sttt/DavidLLNTW15,
  DBLP:journals/sttt/DavidLLMNRSW12}~introduce a specification theory
based on a variant of the timed input/output automata ($\TIOA$)
of~\cite{DBLP:series/synthesis/2010Kaynar,
  DBLP:conf/rtss/KaynarLSV03}.  Both models and specifications are
$\TIOA$ which are action-deterministic and input-enabled; but models are
further restricted using conditions of output urgency and independent
progress.  The equivalence on models being specified is timed
bisimilarity.

In~\cite{DBLP:journals/sttt/DavidLLNTW15} it is shown that $\TIOA$ admit
a conjunction, so they form a meet-semilattice up to~$\equiv$.  The
paper also introduces a composition operation and a quotient, but the
quotient is only shown to satisfy the property that
\begin{equation*}
  \mcal S_1\| \mcal M\le \mcal S_3 \liff \mcal M\le \mcal S_3\by \mcal
  S_1
\end{equation*}
for all specifications $\mcal S_1, \mcal S_3$ and all \emph{models}
$\mcal M$, which is strictly weaker than~\eqref{eq:quotient}.  With
this caveat, $\TIOA$ form a complete compositional semi-logical
specification theory: a bounded residuated semilattice-ordered
commutative semigroup.

\subsection{Abstract probabilistic automata}

Abstract probabilistic automata ($\APA$), introduced
in~\cite{DBLP:journals/iandc/DelahayeKLLPSW13,
  journals/lmcs/DelahayeFLL14}, form a specification theory for
probabilistic automata ($\PA$)~\cite{DBLP:journals/njc/SegalaL95}
under probabilistic bisimilarity.  They build on earlier models of
interval Markov chains
($\IMC$)~\cite{DBLP:journals/jlp/DelahayeLLPW12}, see
also~\cite{DBLP:journals/tcs/BartDFLMT18,
  DBLP:conf/vmcai/DelahayeLP16} for a related line of work.

In~\cite{DBLP:journals/iandc/DelahayeKLLPSW13} it is shown that $\APA$
admit a conjunction, but that $\IMC$ do not.  Also a composition is
introduced in~\cite{DBLP:journals/iandc/DelahayeKLLPSW13}, and it is
shown that composing two $\APA$ with interval constraints (an $\IMC$)
may yield an $\APA$ with \emph{polynomial} constraints (not an
$\IMC$); but $\APA$ with polynomial constraints are closed under
composition.  $\APA$ form a compositional semi-logical specification
theory: a bounded semilattice-ordered commutative semigroup.

\begin{table}[tbp]
  \centering
  \caption{Algebraic taxonomy of some specification theories.
    Abbreviations: L---logical; C---compositional; Q---complete}
  \label{fi:taxo}
  \medskip
  \renewcommand*\arraystretch{1.2}
  \setlength{\tabcolsep}{4pt}
  \begin{tabular}{c|c|c|c|c|p{5cm}}
    Specifications & Models & L & C & Q & Notes \\\hline
    $\HMLR$, $\DMTS$, $\NAA$ & $\LTS$, bisim. & \cmark & \cmark & \cmark &
    \cite{DBLP:journals/iandc/BenesFKLT20}; bisimulation \\
    $\HMLR$, $\DMTS$, $\NAA$ & $\LTS$, any & \xmark & \xmark & \xmark &
    \cite{DBLP:journals/jlap/FahrenbergL20}; any equivalence in
    LTBT spectrum~\cite{inbook/hpa/Glabbeek01} \\
    $\DMTS$ & $\LTS$, fail./div. & \onehalf & \cmark & \xmark &
    \cite{DBLP:journals/tecs/BujtorSV17}; failure/divergence
    equivalence; no disjunction \\
    $\MECS$ & $\ECA$, t.bisim. & \onehalf & \cmark & \cmark &
    \cite{DBLP:journals/scp/BertrandLPR12}; timed bisim.; no
    disjunction \\
    $\TIOA$ & $\TIOA$, t.bisim. & \onehalf & \cmark & \onehalf &
    \cite{DBLP:journals/sttt/DavidLLNTW15}; no
    disjunction; weak quotient \\
    $\IMC$ & $\PA$, p.bisim. & \xmark & \xmark & \xmark &
    \cite{DBLP:journals/jlp/DelahayeLLPW12}; probabilistic bisim. \\
    $\APA$ & $\PA$, p.bisim. & \onehalf & \cmark & \xmark &
    \cite{DBLP:journals/iandc/DelahayeKLLPSW13}; no disjunction
  \end{tabular}
\end{table}

Table~\ref{fi:taxo} sums up the algebraic properties of the different
specification theories we have surveyed here, plus the specification
theory for failure/divergence semantics based on $\DMTS$
from~\cite{DBLP:journals/tecs/BujtorSV17}.

\bibliographystyle{plain}
\bibliography{mybib}

\end{document}